\documentclass[11pt]{article}
\usepackage{graphicx,hyperref,color}
\usepackage[utf8]{inputenc}
\usepackage{marvosym}
\usepackage{cite}
\usepackage{verbatim}
\usepackage{amsmath,amssymb,amsthm}
\usepackage{amsmath,amssymb}
\usepackage{chngcntr}
\usepackage[stable]{footmisc}
\usepackage{float}
\usepackage{url}
\usepackage{xspace}
\usepackage{paralist}
\usepackage{enumerate}
\usepackage{thmtools}
\usepackage{thm-restate}
\usepackage{hyperref}
\usepackage{cleveref}
\usepackage{fullpage}

\definecolor{orange}{rgb}{1,0.5,0}

\DeclareMathOperator{\mst}{\textsc{MST}}
\DeclareMathOperator{\MST}{\mathrm{MST}}

\newtheorem{theorem}{Theorem}

\newtheorem{lemma}[theorem]{Lemma}
\newtheorem{claim}[theorem]{Claim}

\newtheorem{definition}[theorem]{Definition}
\newtheorem{observation}[theorem]{Observation}

\newcounter{fillctr}
\newcommand{\dm}{\mathrm{diam}}
\newcommand{\edm}{\mathrm{ediam}}

\newcommand{\ce}{c(\epsilon)}
\newcommand{\ed}{\epsilon^{-O(d)}}
\newcommand{\dk}{\Delta_{\mathcal{K}}}

\newcommand*\trunc[2]{\lfloor {#1} \rfloor_{#2}}
\newcommand*\rfrac[2]{{}^{#1}\!/_{#2}}

\date{}

\title{Greedy spanners are optimal in doubling metrics}
\author{Glencora Borradaile \and Hung Le \and Christian Wulff-Nilsen}
\begin{document}
\maketitle
\begin{abstract}
We show that the greedy spanner algorithm constructs a $(1+\epsilon)$-spanner of weight $\epsilon^{-O(d)}w(\mst)$ for a point set in metrics of doubling dimension $d$, resolving an open problem posed by Gottlieb~\cite{Gottlieb15}. Our result generalizes the result by Narasimhan and Smid~\cite{NS07} who showed that a point set in $d$-dimension Euclidean space has a $(1+\epsilon)$-spanner of weight at most $\epsilon^{-O(d)}w(\mst)$. Our proof only uses  the packing property of doubling metrics and thus implies a much simpler proof for the same result in Euclidean space. 
\end{abstract}

\section{Introduction}

For a real value $t\ge 1$, a $t$-spanner of an edge-weighted graph $G$ is a subgraph $S$ such that $d_G(x,y)\leq d_S(x,y) \leq t\cdot d_G(x,y)$ where $d_G$ and $d_S$ denote the shortest path distance functions for vertex pairs in $G$ and $S$, respectively. In this work, we study $t$-spanners of a metric space; more precisely, we study $t$-spanners of the corresponding \emph{metric graph} which is the complete graph on the set of points in the metric where the weight of each edge $pq$ denotes the metric distance between $p$ and $q$. When we refer to metric spaces in the following, we assume them to be finite.

Spanners have been used in many applications including distributed systems, communication networks, robotics and more~\cite{ADDJS93,NS07}. In this work, we are interested in $(1+\epsilon)$-spanners (herein referred to simply as \emph{spanners}) of geometric and metric graphs where $\epsilon < 1$ denotes a fixed constant. One way to measure the quality of spanners is by their \emph{lightness} which is the ratio of the weight of spanner's edges  to the weight of a minimum spanning tree of $G$. One prominent application of spanners with constant lightness is in designing faster polynomial time approximation schemes\footnote{A polynomial-time approximation scheme is an algorithm which, for a fixed error parameter $\epsilon$, finds a solution whose value is within $1\pm\epsilon$ of optimal in polynomial time.} (PTAS) for the TSP problem~\cite{RS98,Gottlieb15,Klein05,BDT14,BLW17}. A classical spanner algorithm~\cite{ADDJS93} constructs a $(1+\epsilon)$-spanner of a graph $G$ by considering edges in non-decreasing order of weight and adding the current edge $pq$ to the spanner if there is not already a $p$-to-$q$ path of weight at most $(1+\epsilon)$ times the weight of $pq$; the resulting spanner is called the \emph{greedy $(1+\epsilon)$-spanner} of $G$.

The doubling dimension~\cite{Assouad83,Heinonen01} of a metric space is the smallest $d$ such that every ball of radius $r$ is covered by $2^d$ balls of radius  most $\frac{r}{2}$. In this work, we show that:

\begin{theorem}\label{thm:main}
The greedy $(1+\epsilon)$-spanner of a metric space of doubling dimension $d$ has lightness $\ed$.
\end{theorem}

Geometric spanners have a rich history. In 2-dimensional Euclidean space, $O(1)$-spanners of lightness $O(1)$ have been known since the late 80s~\cite{ADDJS93,LL89}. Das, Heffernan and Narasimhan~\cite{DHN93} sketched an intricate argument showing that $t$-spanners for any fixed $t > 1$ in 3-dimensional Euclidean space have lightness $O(1)$. Their main contribution is an analysis of the \emph{leap-frog property} of the spanner edges found by the greedy algorithm. Later, Das, Narasimhan and  Salowe~\cite{DNS95} sketched a generalization of the proof by Das, Heffernan and Narasimhan~\cite{DHN93} to show that $t$-spanners for any fixed $t > 1$ in  $d$-dimensional Euclidean space have lightness $O(1)$. However, the dependency of the lightness on $t$ and $d$ was not explicitly computed. Rao and Smith~\cite{RS98} redid the analysis of Arya, Das, Mount, Salowe and Smid~\cite{ADMSS95} to show that the constant in the work of Das, Heffernan and Narasimhan~\cite{DHN93} is  $\left(\frac{d}{\epsilon}\right)^{O(d)}$. Narasimhan and Smid~\cite{NS07} devoted a $60$-page chapter of \emph{Geometric Spanner Networks} to give full details of the analysis of the lightness of greedy spanners. They show that greedy spanners for $d$-dimensional Euclidean space have lightness $ \epsilon^{-O(d)}$; their proof heavily relies on the geometry of Euclidean space. Our Theorem~\ref{thm:main} immediately implies a simpler proof for the same lightness bound in Euclidean space; it is well-known that a point set in $d$-dimensional Euclidean metric has doubling dimension $\Theta(d)$. Instead of relying on the leap-frog property as in previous works, which is not easy to analyze in doubling metrics, we only use the simple packing property of the doubling metrics where Euclidean space is a special case.

Spanners in doubling metrics were first considered by Gao, Guibas and Nguyen~\cite{GGN06} who showed that a $n$-point set in doubling dimension $d$ has a spanner of $\ed n$ edges. By analyzing the greedy algorithm, Smid~\cite{Smid09} showed that greedy spanners have $O(n)$ edges and $O(\log n)$ lightness. Beating the $O(\log n)$ lightness  bound of Smid~\cite{Smid09} had been an important open problem until the recent work by Gottlieb~\cite{Gottlieb15}, who showed that a metric of doubling dimension $d$ has a spanner of lightness $\left( \frac{d}{\epsilon}\right)^{O(d)}$. We note that the construction of Gottlieb~\cite{Gottlieb15} is non-greedy, conceptually involved and takes $O(n\log^2 n)$ time.  However, two questions remain open. First, can we design a spanner of lightness $\ed$ to match the bound in Euclidean space? Second, is there a more refined analysis of the greedy algorithm to achieve the bound  $\ed$?  Gottlieb~\cite{Gottlieb15} asked the first question in his paper. The second question was partially addressed by Filtser and Solomon~\cite{FS16}, who showed that greedy spanners in doubling metrics (as well as graph classes closed under edge removal) are \emph{existentially optimal}: if there is a spanner construction of lightness bound $l(\epsilon,d)$, then greedy spanners have lightness $O(l(\epsilon,d))$. Combined with Gottlieb's results~\cite{Gottlieb15}, the existential optimality implies that greedy spanners have $\left( \frac{d}{\epsilon}\right)^{O(d)}$ lightness. In this paper, we resolve both questions affirmatively by presenting a refined and comparatively simple analysis of the greedy algorithm. Our result, in combination with the result of  Filtser and Solomon~\cite{FS16}, implies an $O(n\log n)$ time algorithm to find a spanner of lightness $\ed$.

\subsection{Techniques}

Our analysis is built primarily upon our result~\cite{BLW17} in minor-free graphs which in turn is based on the techniques of Chechick and Wulff-Nilsen~\cite{CW16} for general graphs. We briefly review the $H$-minor-free techniques~\cite{BLW17} here, highlighting new ideas required for doubling metrics. The first step is to reduce the problem on the input graph to the same problem on graphs that have unit-edge-weight minimum spanning trees (MSTs), by rounding small-weight edges and subdividing large-weight edges. Then the greedy algorithm is applied to this slightly modified graph making the construction non-greedy as a whole. However, in the setting of doubling metrics, we cannot use the same simplification since rounding changes the metric. We instead directly analyze the greedy spanner of the input.

To analyze the spanner of an $H$-minor free graph,~\cite{BLW17} uses iterative clustering. Spanner edges are partitioned into $\log \frac{1}{\epsilon}$ sets\footnote{Here, $\log$ denotes the base $2$ logarithm.}, and then the total weight of each set is bounded separately; this induces the $\log \frac{1}{\epsilon}$ factor in the lightness. Each set consists of spanner edges in an exponential scale of many \emph{levels}. First, a non-negative credit $c(\epsilon)$ is assigned to each $\mst$ edge of unit weight; $\ce \log \frac{1}{\epsilon}$ is also the lightness of the spanner. In each level, clusters are constructed iteratively from clusters of the previous level; level-$1$ clusters are constructed directly from the $\mst$. An invariant is maintained that each cluster must have some amount of credit to pay for spanner edges in their level. Credits of level-$1$ clusters are taken directly from $\mst$ edges. Credits of level-$i$ clusters are taken from  credits of clusters of level $i-1$ and $\mst$ edges connecting those lower-level clusters. However, to pay for the spanner edges, level-$i$ clusters cannot take all credits from level-$(i-1)$ clusters. Instead, it is guaranteed that on average, each level-$(i-1)$ cluster has a non-trivial amount of credit left to pay for spanner edges. The minor-free property is then used to argue that each cluster on average must pay for only a constant number of spanner edges in each level. 

Our new argument is also based on our earlier iterative cluster construction. However, we rely on the packing property of doubling metrics (defined below) to show that each cluster needs to pay for a constant number of spanner edges in the same level. This property of doubling metric spanners allows us in fact to simplify the cluster construction that was used for $H$-minor free graphs.

Let $G(V,E)$ be the graph representing a metric of doubling dimension $d$. For each edge $e\in E$, we define the weight function $w(e)$ to be the distance between its endpoints in the metric.  Let $n = |V(G)|$ and $m = |E(G)|$. We directly analyze the spanner produced by the greedy algorithm. For a review of the greedy spanner algorithm, see Appendix~\ref{app:greed}. Let $S$ be the greedy spanner of $G$. Smid~\cite{Smid09} showed (in two pages) that:

\begin{lemma}\label{thm:sparsity} 
$|E(S)| \leq \delta(\epsilon) n$ where $\delta(\epsilon) = \epsilon^{O(-d)}$.
\end{lemma}

The following packing property of doubling metric is well-known (see~\cite{Smid09}):
\begin{lemma}[Packing property] \label{lm:packing-dd} A point set $X$ of a metric of doubling dimension $d$ that is contained in a ball of radius $R$ and for every $x\not= y \in X$, $d(x,y) >  r$, has $|X|\leq \left( \frac{4R}{r}\right)^d$.
\end{lemma}

\section{Assigning credits to $\mst$ edges}

Let $w_0 = \frac{w(\mst)}{n-1}$ be the average weight of an $\mst$ edge. We first bound the total weight of edges that have weight at most $w_0$.

\begin{claim} \label{clm:light-edges}
Let $L_S$ be the set of edges of $S$ of weight at most $w_0$. Then, $w(L_S) \leq 2\delta(\epsilon) w(\mst)$.
\end{claim}
\begin{proof}
$w(L_S) \leq w_0 |L_S| \overset{\mbox{\footnotesize{Lemma~\ref{thm:sparsity}}}}{\leq} w_0\delta(\epsilon) n =  \frac{w(\mst)}{n-1}\delta(\epsilon)n \leq 2\delta(\epsilon) w(\mst)$.
\end{proof}

We now focus on bounding the total weight of edges of weight at least $w_0$ in $S$. We subdivide and allocate credits to $\mst$ edges such that every $\mst$ edge has weight at most $w_0$ and at least $c(\epsilon)w_0$ credits where $\ce$ is a constant that only depends on $\epsilon$ and will be specified later. We will guarantee that the total allocated credit is $O(\ce)w(\mst)$ where $O(\ce)$ is also the lightness of the spanner.  First, we subdivide every $\mst$ edge $e$ of weight more than $w_0$ into $\lceil \frac{w(e)}{w_0} \rceil$ new edges with equal weights summing up to $w(e)$; note that each new edge has weight at most $w_0$. Letting $S'$ be the new graph, we have $w(\mst(S')) = w(\mst)$. We then allocate $c(\epsilon)w_0$ credits to each $\mst$ edge of $S'$. 

\begin{claim} \label{clm:mst-credit}
The total credit allocated to the $\mst$ edges of $S'$ is at most $2c(\epsilon) w(\mst)$.
\end{claim}
\begin{proof}
The total credits assigned to $\mst$ edges of $S'$ is:
\begin{equation} 
\begin{split}
\ce w_0|E(\mst(S'))| &\leq  \ce w_0 \sum_{e \in \mst}\lceil \frac{w(e)}{w_0} \rceil \leq \ce w_0 \left(\sum_{e \in \mst}(\frac{w(e)}{w_0} +1)\right)\\
&= \ce w(\mst) + \ce w_0(n-1) = 2\ce w(\mst)
\end{split}
\end{equation} \qedhere
\end{proof}

\section{Iterative Clustering}

Let $J_0 = \{e \in S', w_0 < w(e) \leq \frac{2w_0}{\epsilon} \}$. We first bound the weight of $J_0$ and pay for edges in $J_0$ separately. The purpose is to simplify the base case in the inductive amortized argument that we present below.

\begin{claim} \label{clm:J0-weight} $w(J_0) \leq \frac{4\delta(\epsilon)}{\epsilon}w(\mst)$.
\end{claim}
\begin{proof}
\begin{equation*}
w(J_0) =  \sum_{e \in J_0, w(e) > w_0} w(e)\overset{\mbox{\footnotesize{Lemma~\ref{thm:sparsity}}}}{\leq} \delta(\epsilon) n \frac{2w_0}{\epsilon}\leq \frac{4\delta(\epsilon)}{\epsilon}w(\mst) \qedhere
\end{equation*}
\end{proof}

Let $I_\epsilon = \lceil \log \frac{1}{\epsilon} \rceil$ and $I_n = \lceil \log n \rceil$. Note that the longest distance between any two vertices in $S'$ is at most $n\cdot w_0$.  We partition the spanner edges (of weight at least $w_0$) of $S'$ into $I_n\cdot I_\epsilon$ sets $\{\Pi_i^j, 0 \leq i \leq I_n-1, 0 \leq j \leq I_\epsilon-1\}$ where each edge $e \in \Pi_i^j$ has weight in the range $(\frac{2^{j}}{\epsilon^i}w_0,\frac{2^{j+1}}{\epsilon^i}w_0]$. For each $ 0 \leq  j \leq I_\epsilon -1$, let
\begin{equation} \label{eq:Sj-def}
S_j = \bigcup_{i=0}^{I_n-1} \Pi_i^j
\end{equation}

\begin{lemma} \label{lm:main} For each $ 0 \leq  j \leq I_\epsilon -1$, there is a set of spanner edges $B$ such that $w(B) \leq \ed\cdot w(\mst)$ and $w(S_j\setminus B) \leq \ed w(\mst)$.
\end{lemma}

It is not hard to see that Lemma~\ref{lm:main} directly implies Theorem~\ref{thm:main}. Thus, we only focus on proving Lemma~\ref{lm:main} for a fixed $j$. We refer to edges of $\Pi_i^j$ as edges in \emph{level $i$} (Equation~\ref{eq:Sj-def}). Let $\ell_i = \frac{2^{j+1}}{\epsilon^i}w_0$. Similar to our analysis for $H$-minor free graphs~\cite{BLW17}, we construct a set of clusters, which are subgraphs of $S'$, for each level and guarantee inductively two diameter-credits invariants:

\begin{description}
\item[DC1] A cluster of level $i$ of diameter $k$ has at least $\ce \cdot\max\{k,\frac{\ell_i}{2}\}$ credits.
\item[DC2] A cluster of level $i$ has diameter  at most $g\ell_i$ for some constant $g > 2$ (specified later).
\end{description}

A cluster of level $i$, say $C_i$, is the union of a subset of clusters in level $i-1$ connected by $\mst$ and level-$i$ spanner edges. Clusters of level $i-1$ are referred to as \emph{$\epsilon$-clusters}. To satisfy DC1, we assign the credits from $\epsilon$-clusters in $C_i$ and the $\mst$ edges connecting the $\epsilon$-clusters to $C_i$. However, we need to group $\epsilon$-clusters in such a way that there are some extra $\epsilon$-clusters whose credits are not needed to maintain DC1 for $C_i$. We will use credits of these extra $\epsilon$-clusters to pay for level-$i$ spanner edges incident to every $\epsilon$-cluster in $C_i$. The credit lower bound $c\ell_i/2$ (DC1) helps us achieve the goal.

To guarantee the diameter-credit invariants for level $0$, we greedily break the $\mst$ into components (level-0 clusters) of diameter at least $\ell_0$ and at most $4\ell_0$. Recall $\ell_0 = 2^{j+1}w_0 \leq \frac{2w_0}{\epsilon}$. To guarantee DC1, we use the credits of $\mst$ edges in the longest path of each cluster. Since the credit of each $\mst$ edge is at least its length, DC1 is satisfied.  Invariant DC2 follows directly from the construction. Note that we have already accounted for the weight of spanner edges of $E_0$ in Claim~\ref{clm:J0-weight}.

\subsection{Constructing higher level clusters}

We construct clusters of level $i$ from the $\epsilon$-clusters of level $i-1$. We assume that the stretch of the spanner is $1 + s\epsilon$ for some constant $s$ (independent of $\epsilon$) that we will pick sufficiently big to make our claims below hold. Furthermore, we assume that $\epsilon$ is bounded from above by a sufficiently small positive constant. We call vertices of $V(S')\setminus V(S)$ \emph{virtual vertices}. We call a cluster \emph{virtual}  of it only contains virtual vertices and \emph{non-virtual} otherwise. Let $\mathcal{K}(\mathcal{C}\epsilon, E_i)$ be the multigraph obtained by taking the subgraph of $G$ consisting of $\epsilon$-clusters and spanner edges in $E_i$ and contracting each $\epsilon$-cluster into a single vertex. Let $\ell = \ell_i$. 

\begin{lemma} \label{lm:K-structure} $\mathcal{K}(\mathcal{C}_\epsilon, E_i)$ is a simple graph of degree $\epsilon^{-O(d)}$.
\end{lemma}
\begin{proof}
We leave the details of the proof that $\mathcal{K}(\mathcal{C}_\epsilon, E_i)$ is simple to Appendix~\ref{app:ommitted}. To show the degree bound, first note that virtual clusters are isolated vertices in $\mathcal{K}(\mathcal{C}_\epsilon, E_i)$  since virtual vertices are subdividing vertices incident to edges of weight at most $w_0$. For each non-virtual $\epsilon$-cluster, we designate a non-virtual vertex to be its center. Assuming w.l.o.g.~that $\epsilon\le\frac 1 4$ and picking $s \geq 12g + 4$, we get the following claim:

\begin{claim} \label{clm:center-dist} Let $C_1,C_2, C_3$ be three $\epsilon$-clusters that have $x_1,x_2,x_3$ as centers. Suppose that $C_2,C_3$ are neighbors of $C_1$ in $\mathcal{K}$. Then, $d_G(x_i,x_j) > \epsilon\ell$ for any $1 \leq  i\not= j \leq 3$.
\end{claim}
\begin{proof}
 
Suppose $y_1y_3$ and $z_1z_2$ are two level-$i$ spanner edges such that $y_1,z_1 \in C_1, z_2\in C_2, y_3\in C_3$. We assume,~w.l.o.g, that  $w(y_1y_3) \leq w(z_1z_2)$. Recall $\ell/2  < w(y_1y_3), w(z_1z_2) \leq \ell$.

We only present the proof showing that $d_G(x_2,x_3) > \epsilon\ell$ since a similar but simpler proof holds for $d_G(x_1,x_2)$ and $d_G(x_1,x_3)$.  We assume that $d_G(x_2,x_3) \leq \epsilon\ell$. Let $Q$ be the $z_1$-to-$z_2$ path that consists of: (i) a shortest $z_1$-to-$y_1$ subpath  in $C_1$, (ii) edge $y_1y_3$, (iii) a shortest $y_3$-to-$x_3$ subpath  in $C_3$, (iv) a shortest $x_3$-to-$x_2$ path in $S$ and (v) a  shortest $x_2$-to-$z_2$ subpath in $C_2$. Since in the greedy spanner, edges are added by increasing weight, by the time $z_1z_2$ is added, $d_S(x_2,x_3) \leq (1+s\epsilon)d_G(x_2,x_3) < (1+s\epsilon) \epsilon\ell$. Thus, we have:
\begin{equation*}
\begin{split}
w(Q) &\leq g\epsilon \ell + w(y_1y_3) + g\epsilon\ell + (1+s \epsilon)\epsilon\ell + g\epsilon \ell\\
&\leq (3g + 1+s\epsilon) \epsilon\ell + w(y_1y_3) \\
&\leq 2(3g + 1+s\epsilon) \epsilon w(z_1z_2) + w(z_1z_2) \\
&= \left((6g + 2 + 2s\epsilon)\epsilon + 1\right) w(z_1z_2) 
\end{split}
\end{equation*}
Since $\epsilon\le\frac 1 4$, we have $s \geq 12g + 4\geq (6g + 2)/(1-2\epsilon)$ and hence $6g + 2 + 2s\epsilon\le s$. But then $w(Q) \leq (1 + s\cdot\epsilon)w(z_1z_2)$; contradicting that $z_1z_2$ is a spanner edge.
\end{proof}

Let $C_0$ be an $\epsilon$-cluster with neighbors $C_1, C_2, \ldots, C_p$ in $\mathcal{K}$. Let $X = \{x_0,x_1,\ldots, x_p\}$ where $x_i$ is the center of $C_i, 0 \leq i \leq p$. We show below that $d_G(x_0,x_i) \leq 3\ell$ for every $1 \leq i \leq p$ when $\epsilon$ is sufficiently small. Thus, $X$ is contained in a ball centered at $x_0$ of radius at most $3\ell$. By Claim~\ref{clm:center-dist}, $d_G(x_i, x_j) > \epsilon \ell$ for every $0 \leq  i \leq p$. Thus, by Lemma~\ref{lm:packing-dd}, $|X| \leq \ed$.

We now show that $d_G(x_0,x_i) \leq 3\ell$ for a fixed $i$ with $1\le i\le p$. Let $y_0y_i$ be the spanner edge in $E_i$ connecting $C_0$ and $C_i$ in $\mathcal{K}$. Then the $x_0$-to-$x_i$ path $P$ consisting of an $x_0$-to-$y_0$ shortest path in $C_0$, edge $y_0y_i$ and a $y_i$-to-$x_i$ shortest path in $C_i$ has length at most $2g\epsilon \ell + \ell \leq 3\ell$ when $\epsilon$ is smaller than $1/g$. 
\end{proof}

Let $\mathcal{T}$ be a tree of $\epsilon$-clusters connected by $\mst$ edges. We say an $\epsilon$-cluster of $\mathcal{T}$ is \emph{branching} if it is incident to at least three $\mst$ edges in $\mathcal{T}$. Let $\mathcal{P}$ be a path of $\mathcal{T}$. We define the diameter of $\mathcal{P}$, denoted by $\dm(\mathcal{P})$, to be the diameter of the subgraph of $S'$ formed by edges inside $\epsilon$-clusters  and $\mst$ edges connecting $\epsilon$-clusters of $\mathcal{P}$.  We define \emph{effective diameter} of $\mathcal{P}$, denoted by $\edm(\mathcal{P})$, to be the diameters of $\epsilon$-clusters in $\mathcal{P}$. Since $\epsilon$-clusters have diameter at least $w_0$ (by construction of the base case) which is at least the weight of edges connecting them in $\mathcal{P}$, we have:
\begin{observation} \label{obs:dm-vs-edm}$\dm(\mathcal{P}) \leq 2\edm(\mathcal{P})$.
\end{observation}
We define the effective diameter of a subtree $\mathcal{T}'$ of $\mathcal{T}$ to be the effective diameter of the diameter path of $\mathcal{T}'$. We construct clusters in four phases:

\paragraph{Phase 1: Branching $\epsilon$-clusters.} We have two steps. Since our construction is recursive, we update the set branching vertices of $\mathcal{T}$ after each recursive step. (Step 1) Let $\mathcal{T}'$ be a subtree of $\mathcal{T}$ of effective diameter at least $\ell$ and at most $2\ell$ that contains a branching vertex $X$ and its neighbors so that $X$ is still branching in $\mathcal{T}'$.  We group $\epsilon$-clusters and $\mst$ edges of $\mathcal{T}'$ as a new level-$i$ cluster. We remove $\mathcal{T}'$ from $\mathcal{T}$ and repeat until every component of $\mathcal{T}$ either has effective diameter less than $\ell$ or is a path of $\epsilon$-clusters, called a \emph{cluster path}, of effective diameter at least $\ell$. 

(Step 2) Let $\mathcal{P}$ be a cluster path of diameter at least $\ell$. Let $X$ be the set of internal $\epsilon$-clusters of $\mathcal{P}$ such that $X$ has at least one $\mst$ edge, say $e$, to a level-$i$ cluster, say $C$, formed in Step 1. Observe that $X$ is branching before execution of Step 1 and the removal of subtrees of $\mathcal{T}$ in Step 1 reduces degree of $X$ to 2. We augment $C$ with $X$ and $e$. We then remove $X$ from $\mathcal{P}$ and repeat until every cluster path of effective diameter at least $\ell$ only has $\mst$ edges to level-$i$ clusters incident to its endpoint $\epsilon$-clusters. 

\paragraph{Phase 2:  $\epsilon$-clusters in high diameter paths.} Let $e$ be a spanner edge in $E_i$ whose endpoints, $x$ and $y$ are in high-diameter cluster paths, $\mathcal{P}$ and $\mathcal{Q}$, respectively, where it may be that $\mathcal{P} = \mathcal{Q}$.  Let $C_x$ and $C_y$ be the $\epsilon$-clusters containing $x$ and $y$, respectively.  We only proceed with this phase if the two affix cluster subpaths of $\mathcal{P}$ ending at $C_x$ have effective diameter at least $\ell$ (likewise for $\mathcal{Q}$).  Let $\mathcal{P}_1$  and $\mathcal{P}_2$ be the two minimal subpaths of $\mathcal{P}$ ending at $C_x$ that have effective diameter at least $\ell$.  Likewise define $\mathcal{Q}_1$ and $\mathcal{Q}_2$.  We group $\epsilon$-clusters and $\mst$ edges of $\mathcal{P}_1\cup \mathcal{Q}_1 \cup \mathcal{P}_2 \cup \mathcal{Q}_2$ and $e$ as a new level-$i$ cluster.  See Figure~\ref{fig:P2} for an illustration of the different forms this cluster can take. 
\begin{figure}
\centering
\includegraphics[scale = 1]{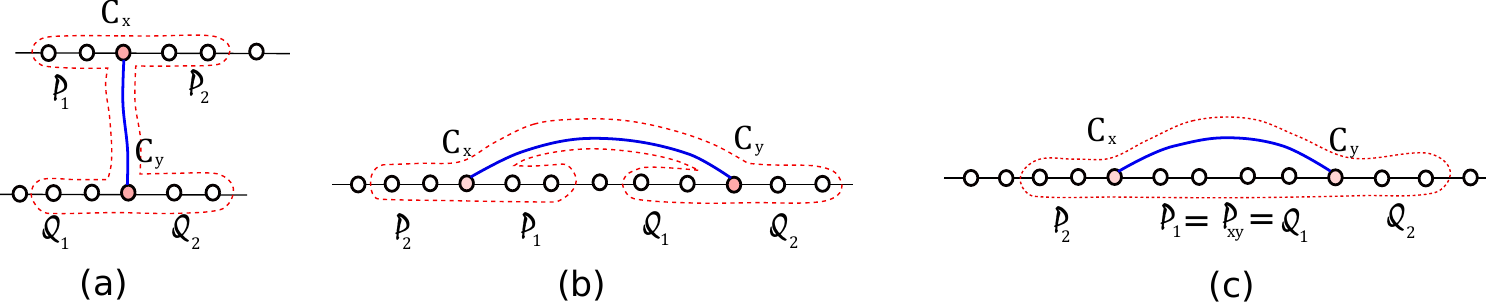}
\caption{Three different forms that a cluster (enclosed in the dotted red curves) in Phase 2 can take. The solid blue line is the spanner edge $e$. (a) $e$ connects $\epsilon$-clusters in different clusters paths, (b) $e$ connects $\epsilon$-cluster in the same path and $\mathcal{P}_1$  and $\mathcal{Q}_1$ are disjoint and (c) $e$ connects $\epsilon$-cluster in the same path and $\mathcal{P}_1$  and $\mathcal{Q}_1$ are overlapped. In case (c), we redefine $\mathcal{P}_1 = \mathcal{Q}_1 = \mathcal{P}_{xy}$.}
\label{fig:P2}
\end{figure}

\paragraph{Phase 3: Low diameter components.} Let $\mathcal{F}$ be the set of trees (and paths) remaining of effective diameter at most $\ell$. By construction, each component $\mathcal{T}'$ of $\mathcal{F}$ has a $\mst$ edge, say $e$, to a level-$i$ cluster constructed in previous phases, say $C$. We augment $C$ by $\mathcal{T}'$ and $e$.

\paragraph{Phase 4: Remaining high diameter paths.} Let $\mathcal{P}$ be a cluster path of effective diameter at least $\ell$. We greedily break $\mathcal{P}$ into subpaths of effective diameter at least $\ell$ and at most $2\ell$. If any affix of $\mathcal{P}$, say $\mathcal{P}'$, has a $\mst$ edge, say $e$, to a level-$i$ cluster constructed in previous phases, say $C$,  we augment $C$ with $\mathcal{P}'$ and $e$. We then  make each remaining cluster subpath of $\mathcal{P}$ into a new level-$i$ cluster. 

This completes the cluster construction for level $i$.

\subsection{Showing diameter-credit invariant DC2}
By construction, each level-$i$ cluster constructed in Phase 4 is a cluster path of effective diameter at most $2\ell$. By Observation~\ref{obs:dm-vs-edm}, we have:
\begin{claim} \label{clm:DC2-P4} Level-$i$ clusters constructed in Phase 4 have diameter at most $4\ell$.
\end{claim}

\begin{claim} \label{clm:dm-Cp}  Level-$i$ clusters have diameter at most $33\ell $ when $\epsilon$ is smaller than $\frac{1}{g}$. 
\end{claim}
\begin{proof}
Let $C$ be a level-$i$ cluster that is initially formed in Phase 1 or 2.  By construction, $C$ may be augmented in Phases 3 and 4. Let $C'$ and $C''$ be the augmented clusters of $C$ after Phase 3 and Phase 4, respectively. It could be that $C = C' = C''$. $C'$ is obtained from $C$ by attaching trees of effective diameter at most $\ell$ via $\mst$ edges. $C''$ is obtained from $C'$ by attaching trees of effective diameter at most $2\ell$ via $\mst$ edges. Recall each $\mst$ edge has length at most $w_0$. By Observation~\ref{obs:dm-vs-edm}, we have:
\begin{equation} \label{eq:C-vs-Cp-vs-Cpp}
\dm(C') \leq \dm(C) + 4\ell + 2w_0 \qquad \mbox{and} \qquad \dm(C'') \leq \dm(C') + 8\ell+2w_0
\end{equation}

If $C$ is constructed in Phase 1, by Observation~\ref{obs:dm-vs-edm}, after Step 1, $\dm(C) \leq 4\ell$. Since in Step 2, $C$ is augmented by $\epsilon$-clusters via $\mst$ edges, after Step 2, $\dm(C) \leq 4\ell + 2w_0 + 2g\epsilon\ell \leq 8\ell$ ($\ell \geq w_0$ by construction of the base case). If $C$ is constructed in Phase 2, we have:
\begin{equation*}
\dm(C) \leq \dm(\mathcal{P}_1) + \dm(\mathcal{P}_2) + \dm(\mathcal{Q}_1) + \dm(\mathcal{Q}_2) + \ell(e)
\end{equation*}
Since $\mathcal{P}_1, \mathcal{P}_2, \mathcal{Q}_1, \mathcal{Q}_2$ are minimal, each has effective diameter at most $\ell + g\epsilon\ell$. Thus, $\dm(C) \leq 4(2\ell + 2g\epsilon\ell) + \ell = 9\ell + 9g\epsilon\ell \leq 17\ell$. Thus, in both cases, $\dm(C) \leq 17\ell$. By Equation~\eqref{eq:C-vs-Cp-vs-Cpp}, $\dm(C'') \leq 29\ell + 4w_0 \leq 33\ell$.
\end{proof}

Thus, by Claim~\ref{clm:dm-Cp}, we can choose $g = 33$.

\subsection{Showing diameter-credit invariant DC1}

Let $\dk$ be the maximum degree of the cluster graph $\mathcal{K}$. By Lemma~\ref{lm:K-structure}, $\dk = \ed$. We define $cr(\mathcal{X})$ to be the total credits of a set of $\epsilon$-clusters $\mathcal{X}$. 
\subsubsection{Clusters originating in Phase 4}

Let $C$ be a level-$i$ cluster formed in Phase 4. We call $C$  a \emph{long cluster} if it has at least $\frac{2g}{\epsilon} + 1$ $\epsilon$-clusters and a \emph{short cluster} otherwise. We have:
\begin{claim} \label{clm:DC1-long-cluster}
A long cluster can both maintain invariant DC1 and pay for its incident spanner edges when $\ce = \ed$.
\end{claim}
\begin{proof}
Let $\mathcal{X}$ be a set of any $\frac{2g}{\epsilon}$ $\epsilon$-clusters of $C$. By invariant DC1 for level $i-1$, we have:

\begin{equation*}
cr(\mathcal{X}) \geq  \frac{2g}{\epsilon} \ce\ell/2 = \ce g\ell 
\end{equation*}
which is at least $\ce \cdot\max(\dm(C), \ell/2)$ since $\dm(C) \leq g\ell$ as shown in Claim~\ref{clm:DC2-P4} (since $g = 33$).  Thus, credits of $\mathcal{X}$ are enough to maintain DC1 for $C$. 

Since $C$ is a long cluster, there is at least one $\epsilon$-cluster, say $Y$, not in $\mathcal{X}$. By DC1 for level $i-1$, $Y$ has at least $\ce\ell/2$ credits. Since there are at most:
\begin{equation*}
\dk\cdot \left(\frac{2g}{\epsilon} + 1\right) = \ed
\end{equation*}
level-$i$ spanner edges incident to $\epsilon$-clusters in $\mathcal{X}\cup\{Y\}$, $Y$'s credits are enough to pay for those spanner edges when $\ce = \ed$.  

For each $\epsilon$-cluster $z \in C\setminus (\mathcal{X}\cup\{Y\})$, we use $z$'s credit to pay for the spanner edges incident to $z$. By Lemma~\ref{lm:K-structure} and invariant DC1 for level $i-1$, this amount of credit is sufficient when $\ce = \ed$.
\end{proof}

\begin{claim} \label{clm:DC1-int-short} The credits of $\epsilon$-clusters and $\mst$ edges connecting $\epsilon$-clusters of each short cluster $C$ are enough to maintain invariant DC1 for $C$.
\end{claim}
\begin{proof}
We abuse notation by letting $\mst(C)$ be the set of $\mst$ edges in $C$ that connects  its $\epsilon$-clusters. Since $C$ is a cluster path, we have: 
\begin{equation*}
\begin{split}
\dm(C) &\leq \sum_{X_\epsilon \in C}\dm(X_\epsilon) +  \sum_{e \in \mst(C)} w(e)\\
\end{split}
\end{equation*}
By invariant DC1 for level $i-1$, $cr(X_\epsilon) \geq \ce \cdot \dm(X_\epsilon)$ and since each MST edge has credit at least $\ce$ times its length, the claim follows.
\end{proof}

A short cluster may need to use all the credits of $\epsilon$-clusters and $\mst$ edges to maintain DC1, hence, it many not have extra credit to pay for any incident level-$i$ spanner edges. In this case, we need to use credits of other level-$i$ clusters to pay for those spanner edges. We call a short cluster \emph{internal} if it is not an affix of a long path $\mathcal{P}$ in Phase 4. 
\begin{observation} \label{obs:edge-int-short}
There is no level-$i$ spanner edge $e$ that has both endpoints in internally short clusters. 
\end{observation}
\begin{proof}
If there is such an edge $e$, it would be grouped into a level-$i$ cluster in Phase 2. 
\end{proof}

Thus, a level-$i$ spanner edge incident to an internally short cluster can be paid by the level-$i$ cluster that contains the other endpoint of $e$. However, if a short cluster is not internal, we must find a way to pay for its incident spanner edges. Recall after Phase 1, every cluster path of effective diameter at least $\ell$ must have an $\mst$ edge from one of its endpoint $\epsilon$-clusters to a level-$i$ cluster. By construction in Phase 4, if a short cluster is an affix of $\mathcal{P}$, called a \emph{short affix cluster},  the other affix of $\mathcal{P}$, called the \emph{sibling affix}, must have an $\mst$ edge to a cluster originating in the first two phases and thus augments it. (The only exception is when there is no level-$i$ clusters after the first two phases and we will handle this case at the end of this paper.) Thus, we can use the credit of $\epsilon$-clusters of the sibling affices to pay for incident spanner edges of affix short clusters. To that end, we analyze clusters originally constructed in the first two phases.

\subsubsection{Clusters originating in Phase 1+2}

Let $C$ be a level-$i$ cluster constructed in Phase 1 or 2. Let $C'$ and $C''$ be the augmentations of $C$ in Phase  3 and 4, respectively. Let $D$ be the diameter path of the spanner given by edges and vertices in $C''$. Let $\mathcal{D}$ be the walk obtained from $D$ by contracting each maximal subpath of $D$ that is inside an $\epsilon$-cluster of $C''$ to a single vertex. 

\begin{definition}[Canonical pair] Let $\mathcal{S} \subseteq C \cup \mathcal{D}$ be a subset of $\epsilon$-clusters of $C''$ such that $|\mathcal S| \leq \frac{2g}{\epsilon}$ and the credits of $\epsilon$-clusters in $\mathcal{S}$ and $\mst$ edges of $C''$ are sufficient to maintain invariant DC1 for $C''$. Let $Y$ be an $\epsilon$-cluster of $C$ that is not in $\mathcal S$. We call $(\mathcal{S},Y)$ a \emph{canonical pair} of $C''$.
\end{definition}

Note that we do not claim the existence of canonical pairs. Indeed, the main goal of this subsection is to prove that a canonical pair exists for $C''$ since its existence implies that $C''\setminus \mathcal{S} \not= \emptyset$. Thus, we can use credits of $\epsilon$-clusters in $C''\setminus \mathcal{S}$ to pay for level-$i$ spanner edges incident to $\epsilon$-clusters of $C''$ and $\epsilon$-clusters of short affix clusters hat have sibling affices in $C''$.

\begin{claim}\label{clm:RSY-pay-spanner-edge}
If $C''$ has a canonical pair $(\mathcal{S}, Y)$, we then can pay for every level-$i$ spanner edge that is incident to $\epsilon$-clusters of $C''$ and $\epsilon$-clusters of short affix clusters that have sibling affices in $C''$ using credits of $\epsilon$-clusters in $C''\setminus \mathcal{S}$ when $\ce = \epsilon^{-\Theta(d)}$.
\end{claim}
\begin{proof}
 Let $\mathcal{R}$ be a set of $\epsilon$-clusters that contains every $\epsilon$-cluster in $\mathcal{S} \cup \{Y\}$ and  affix short clusters in Phase 4 whose sibling affices contain an $\epsilon$-cluster of $\mathcal{D}$. Recall that $C'$ is augmented by attaching paths of $\epsilon$-clusters via $\mst$ edges. Thus, $\mathcal{R}$ contains at most two short clusters  as a result of Phase 4 (see Figure~\ref{fig:canonical-pair}).  Since $|\mathcal{S}| \leq \frac{2g}{\epsilon}$ and each short cluster has at most $\frac{2g}{\epsilon}$ $\epsilon$-clusters, $|\mathcal{R}| = O(\frac{g}{\epsilon})$. Since each $\epsilon$-cluster is incident to at most $\dk$ level-$i$ spanner edges by Lemma~\ref{lm:K-structure}, $\epsilon$-clusters in $\mathcal{R}$ are incident to at most $O(\frac{g\dk}{\epsilon})$  level-$i$ spanner edges. Recall that each level-$i$ spanner edge has length at most $\ell$. By invariant DC1 for level $i-1$, $Y$ has at least $\ce \epsilon\ell/2$ credits. Thus, by choosing $\ce = \Theta(\frac{g\dk}{\epsilon^2}) = \epsilon^{-\Theta(d)}$, $Y$'s credit is sufficient to pay for every spanner edge incident to $\epsilon$-clusters in $\mathcal{R}$.

\begin{figure}
\centering
\includegraphics[scale = 1.5]{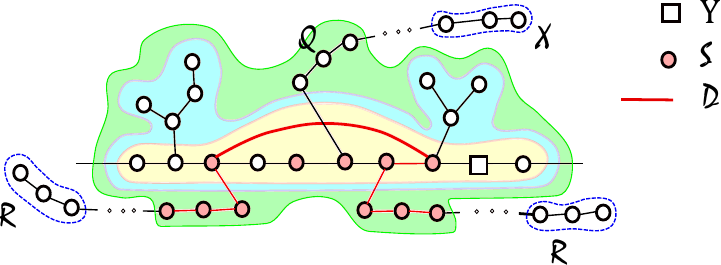}
\caption{ Clusters $C, C'$ and $C''$ are enclosed  by yellow-shaded, cyan-shaded and green-shaded regions, respectively. The red path is the cluster walk $\mathcal{D}$. Shaded $\epsilon$-clusters are in $\mathcal{S}$ and $Y$ is the square $\epsilon$-cluster. Short affix clusters in Phase 4 are enclosed by dotted blue curves. $\mathcal{R}$ contains $\mathcal{S}\cup \{Y\}$ and two (annotated) short affix clusters that have two corresponding sibling affices in $\mathcal{D}$.}
\label{fig:canonical-pair}
\end{figure}
For each $\epsilon$-cluster  $z$ in $C''\setminus \mathcal{R}$, the credit of $z$ is sufficient to pay for incident level-$i$ spanner edges incident to $z$. However, we also need to pay for short affix clusters in Phase 4, whose siblings augment $C'$ in Phase 4. To afford this, we use half the credit of each $\epsilon$-cluster in $C\setminus \mathcal{R}$ (of value at least $\ce \epsilon\ell/4$ by invariant DC1 for level $i-1$) to pay for level-$i$ spanner edges incident to it. Since each $\epsilon$-cluster is incident to at most  $\dk$ level-$i$ spanner edges by Lemma~\ref{lm:K-structure}, this credit is sufficient when $\ce \geq \frac{4\dk}{\epsilon} = \ed$.

Since $C'$ is augmented by attaching cluster paths via $\mst$ edges, an affix cluster not in $\mathcal{R}$ has its sibling in a subset of $C''\setminus \mathcal{R}$. For each short affix cluster $\mathcal{X}$ (see Figure~\ref{fig:canonical-pair}) whose sibling affix, say $\mathcal{Q}$, is in $C''$, we use the remaining half of the credits of the $\epsilon$-clusters of $\mathcal{Q}$ to pay for the level-$i$ spanner edges incident to $\mathcal{X}$. By Lemma~\ref{lm:K-structure}, $X$ is incident to at most $\frac{2g}{\epsilon}\dk$ level-$i$ spanner edges. Since $\edm(\mathcal{Q}) \geq \ell$, $cr(\mathcal{Q}) \geq \ce\ell$ by invariant DC1 for level $i-1$. Thus, half the credit of $\mathcal{Q}$ is sufficient when $\ce \geq \frac{2g}{\epsilon}\dk = \epsilon^{-O(d)}$. 
\end{proof}

By Claim~\ref{clm:RSY-pay-spanner-edge}, it remains to show that $C''$ has a canonical pair $(\mathcal{S}, Y)$. Let $\mathcal{X}$ be a set of $\epsilon$-clusters. We define a subset of $\mathcal{X}$ as follows:
\[
 \trunc{\mathcal{X}}{\rfrac{2g}{\epsilon}} = 
  \begin{cases} 
   \mathcal{X} & \text{if } |\mathcal{X}| \leq \rfrac{2g}{\epsilon}\\
   \text{any subset of }  \rfrac{2g}{\epsilon} \epsilon\text{-clusters of } \mathcal{X}      & \text{otherwise }
  \end{cases}
\]

\begin{claim} \label{clm:DC1-P1} If $C$ is constructed in Phase 1, then $C''$ has a canonical pair.
\end{claim}
\begin{proof}
Recall $C$ is a tree of $\epsilon$-clusters. That implies $C''$ is also a tree of $\epsilon$-clusters that are connected by $\mst$  edges. Thus, $\mathcal{D}$ is a simple path.  Since $C$ contains a branching $\epsilon$-cluster $X$, there must be at least one neighbor $\epsilon$-cluster of $X$ that is not in $\mathcal{D}$. Let $Y$ be an arbitrary neighbor  $\epsilon$-cluster in $C$ of $X$ and  $\mathcal{S} = \trunc{\mathcal{D}}{\rfrac{2g}{\epsilon}}$.  By definition, $|\mathcal{S}| \leq \frac{2g}{\epsilon}$. 

It remains to show that credits of $\epsilon$-clusters of $\mathcal{S}$ and $\mst$ edges of $D$ is  sufficient to guarantee invariant DC1 for $C''$. Suppose that $\mathcal{S}$ contains at least $\frac{2g}{\epsilon}$ $\epsilon$-clusters. By invariant DC1 for level $i-1$, $cr(S) \geq \ce g\ell \geq \ce \max(\dm(C''),\ell/2)$ which is enough to maintain invariant DC1. Thus, we can assume that $\mathcal{S}$ contains every $\epsilon$-cluster of $\mathcal{D}$. Since $\mathcal{D}$ consists of $\epsilon$-clusters and $\mst$ edges only, we have:
\begin{equation*}
\dm(\mathcal{D}) \leq \sum_{X_\epsilon \in \mathcal{D}}\dm(X_\epsilon) + \sum_{e\in \mst(\mathcal{D})} w(e)
\end{equation*}
Thus, credits of $\epsilon$-clusters of $\mathcal{S}$ and $\mst$ edges of $\mathcal{D}$ are sufficient to maintain DC1. 
\end{proof}
We now consider the case when $C$ is constructed in Phase 2. Recall $C$ consists of four paths $\mathcal{P}_1,\mathcal{P}_2, \mathcal{Q}_1, \mathcal{Q}_2$ that are not necessarily distinct and a single spanner edge $e$ (see Figure~\ref{fig:P2}). 

\begin{claim} \label{clm:DC1-P2-1} If  the four paths $\mathcal{P}_1,\mathcal{P}_2, \mathcal{Q}_1, \mathcal{Q}_2$  are distinct, then $C''$ has a canonical pair.
\end{claim}
\begin{proof}
Let $\mathcal{F} = \{\mathcal{P}_1,\mathcal{P}_2, \mathcal{Q}_1, \mathcal{Q}_2\}$. By construction in Phase 2, $C$ is an acyclic graph of $\epsilon$-clusters connected by $\mst$ edges and a single spanner edge $e$. Thus, $\mathcal{D}$ is a simple path. That implies at most two paths, say $\mathcal{P}'$ and $\mathcal{Q}'$, among four paths in $\mathcal{F}$ share $\epsilon$-clusters with $\mathcal{D}$. Let other two paths of $\mathcal{F}$ be $\mathcal{P}''$ and $\mathcal{Q}''$.  Let $Y$ be an arbitrary $\epsilon$-cluster of $\mathcal{Q}''$ and

\begin{equation*}
\mathcal{S} = \trunc{\mathcal{D} \cup \mathcal{P}' \cup  \mathcal{Q}' \cup  \mathcal{P}''}{\rfrac{2g}{\epsilon}}
\end{equation*} 
If $\mathcal{S} = \frac{2g}{\epsilon}$, then $cr(\mathcal{S}) \geq \ce g \ell$ by invariant DC1 for level $i-1$. Hence, credits of $\epsilon$-clusters in  $\mathcal{S}$ are sufficient to maintain DC1 for $C''$ since $\dm(C'')\leq g\ell$ as shown in the previous section. 

Thus, we can assume that $\mathcal{S} < \frac{2g}{\epsilon}$. In this case, $\mathcal{S} = \mathcal{D} \cup \mathcal{P}' \cup  \mathcal{Q}' \cup  \mathcal{P}'' $.  If $\mathcal{D}$ does not contain the spanner edge $e$, using the same argument in Claim~\ref{clm:DC1-P1}, we can show that credits of $\epsilon$-clusters and $\mst$ edges of $\mathcal{D}$ are sufficient to maintain invariant DC1 for $C''$. Otherwise, we assign credits of $\mathcal{P}''$ to $e$. Since $\edm(\mathcal{P}'') \geq \ell \geq w(e)$, by invariant DC1 for level $i-1$, $cr(\mathcal{P}'') \geq \ce \edm(\mathcal{P}'') \geq \ce w(e)$. Thus $e$ is assigned credit of at least $\ce$ times its length. We then use credits of $\epsilon$-clusters and edges of $\mathcal{D}$ to maintain DC1. The rest of the proof is similar to Claim~\ref{clm:DC1-P1}.
\end{proof}

We assume that $\mathcal{P}_1 = \mathcal{Q}_1 = \mathcal{P}_{xy}$.  In this case, $C$ contains a unique cycle, which is $\{e\} \cup   \mathcal{P}_{xy}$. We first prove that $\mathcal{D}$ is a path when $\epsilon$ is sufficiently small.

\begin{claim} \label{clm:D-simple} $\mathcal{D}$ is a path if $\epsilon$ is smaller than $\frac{1}{2g}$. 
\end{claim}
\begin{proof} If $\mathcal{D}$ is not simple, it contains a cycle $\mathcal{C}_{xy}$. Let $u$ and $v$ be two vertices of the same $\epsilon$-cluster, say $X_\epsilon$, such that $D$ enters and leaves $\mathcal{C}_{xy}$ at $u$ and $v$, respectively. Then, the subpath $D_{uv}$ between $u$ and $v$ of $D$ must contain edge $e$ of length at least $\ell/2$. However, we can shortcut $D_{uv}$ through $X_\epsilon$ by a path of length at most $\dm(X_\epsilon) \leq g\epsilon \ell$ by DC2. For $\epsilon < \frac{1}{2g}$, the shortcut has length smaller than $w(D_{uv})$, contradicting that $D$ is a shortest path.  
\end{proof}

\begin{observation} \label{obs:Pxy-no-in-D}
   $\mathcal{P}_{xy} \not\subseteq \mathcal{D}$.
\end{observation}

\begin{proof}
  For otherwise, $\mathcal{D}$ could be shortcut through $e$ at a cost of
  \begin{eqnarray*}
    & \le&  \underbrace{\dm(C_x) + \dm(C_y) + w(e)}_\text{cost of shortcut}  - \underbrace{(\dm(\mathcal{P}_{xy}) - \dm(C_x) - \dm(C_y))}_\text{lower bound on diameter} \\
    & \le&  w(e) +4g\epsilon\ell -(1+s\epsilon)w(e) \quad \text{(by the stretch condition for }e)\\
    & \le& 4g\epsilon\ell -s\epsilon\ell/2 \quad \text{(since }w(e) \geq \ell/2)
  \end{eqnarray*}
 This change in cost is negative for $s \geq 8g +1$. 
\end{proof}

\begin{claim} \label{clm:DC1-P2-2} $C''$ has a canonical pair.
\end{claim}
\begin{proof}
Let $Y$ be an $\epsilon$-cluster of $\mathcal{P}_{xy}\setminus \mathcal{D}$. $Y$ exists by Observation~\ref{obs:Pxy-no-in-D}. We define:
\begin{equation*}
\mathcal{S} =  \trunc{\mathcal{D}\cup \mathcal{P}_2 \cup \mathcal{Q}_2 \cup \mathcal{P}_{xy} \setminus \{Y\}}{\rfrac{2g}{\epsilon}}
\end{equation*}
If $|\mathcal{S}| = \frac{2g}{\epsilon}$, then the total credit of $\epsilon$-clusters in $\mathcal{S}$ is at least $\ce g\ell$ by invariant DC1 for level $i-1$. Thus credits of $\epsilon$-clusters in $\mathcal{S}$ is sufficient to maintain invariant DC1 for $C''$. That implies $C''$ has a canonical pair.

Otherwise, $\mathcal{S} = \mathcal{D}\cup \mathcal{P}_2 \cup \mathcal{Q}_2 \cup \mathcal{P}_{xy} \setminus \{Y\}$. If $\mathcal{D}$ does not contain the spanner edge $e$, then by the same argument in Claim~\ref{clm:DC1-P1}, we can argue that credits of $\epsilon$-clusters and $\mst$ edges in $\mathcal{D}$ is enough to maintain invariant DC1 for $C''$. Thus, we can assume that $\mathcal{D}$ contains $e$. We consider two cases:
\begin{enumerate}
\item $\mathcal{D}$ contains an internal $\epsilon$-cluster of $\mathcal{P}_{xy}$. Since $\mathcal{D}$ is a path by Claim~\ref{clm:D-simple}, its does not contain any internal $\epsilon$-cluster of at least one of two paths $\mathcal{P}_2, \mathcal{Q}_2$,~w.l.o.g., say $\mathcal{P}_2$. Since $\edm(\mathcal{P}_2) \geq \ell$, by invariant DC1 for level $i-1$, the total credit of $\epsilon$-clusters in $\mathcal{P}_2$ is at least $\ce \ell$ which is at least $\ce w(e)$. Thus, by assigning credits of $\mathcal{P}_2$ to $e$, every edge of $\mathcal{D}$ has credit at least $c(\epsilon)$ times it length. Thus, credits of $\epsilon$-clusters and edges of $\mathcal{D}$ are enough to maintain DC1 for $C''$.

\item $\mathcal{D}$ does not contain any internal $\epsilon$-cluster of $\mathcal{P}_{xy}$. We have:
\begin{equation} \label{eq:dm-L-prime-case2}
\begin{split}
& \dm(\mathcal{P}_{xy} \setminus \{C_x,C_y\}) \\
& \ge \dm(\mathcal{P}_{xy}) - \dm(C_x)-\dm(C_y) \\
& \ge (1+ s\epsilon)w(e) -  \dm(C_x)-\dm(C_y) \ \ \ (\mbox{by the stretch condition}) \\
& \ge w(e) + s\epsilon\ell/2 - 2g\epsilon\ell\ \ \ (\mbox{by bounds on $w(e)$ and DC2}) \\
& \ge w(e) + g\epsilon\ell \ \ \ (\mbox{for $s \ge 8g + 1$, as previously required})
\end{split} 
\end{equation}
The credit of the $\mst$ edges and $\epsilon$-clusters of $\mathcal{P}_{xy} \setminus \{C_x,C_y\}$ is at least:
\begin{equation}
\begin{split}
&\ce \cdot(\mst(\mathcal{P}_{xy} \setminus \{C_x,C_y\}) + \edm(\mathcal{P}_{xy} \setminus \{C_x,C_y\})) \\&\ge \ce\cdot \dm(\mathcal{P}_{xy} \setminus \{C_x,C_y\})\\&\geq \ce (w(e) +g\epsilon\ell)  \qedhere
\end{split}
\end{equation}
Since $\dm(Y) \leq g\epsilon \ell$ by invariant DC2 for level $i-1$, the total credit of $\epsilon$-clusters of $\mathcal{P}_{xy} \setminus \{C_x,C_y,Y\}$ and $\mst$ edges of  $\mathcal{P}_{xy} \setminus \{C_x,C_y\}$ is at least $\ce \cdot w(e)$. Thus, by assigning this credit to $e$, we can argue that credits of $\epsilon$-clusters and edges of $\mathcal{D}$ are enough to maintain DC1 for $C''$.
\end{enumerate}
\end{proof}

\subsubsection{No Phase 1 or 2 clusters}

We now deal with the case when there are no level-$i$ clusters formed in Phase 1 and 2. 

\begin{observation} \label{obs:exception}There is no level-$i$ cluster formed in Phase 1 and 2 if and only if (i) the tree $\mathcal{T}$ of $\epsilon$-clusters is a path and (ii) every spanner edge is incident to an $\epsilon$-cluster in an affix of $\mathcal{T}$ having effective diameter at most $2\ell$.
\end{observation}

By Claim~\ref{clm:DC1-long-cluster}, we only need to pay for spanner edges incident to short affix clusters of $\mathcal{T}$. Since short clusters have at most $\frac{2g}{\epsilon}$ $\epsilon$-clusters, there are at most $\frac{4g\dk}{\epsilon}$ such spanner edges, that we assign to set $B$ (Lemma~\ref{lm:main}). Below, we show that $w(B) \leq \epsilon^{-O(d)}\cdot w(\mst)$ across all levels, implying Lemma~\ref{lm:main}.

\begin{claim} \label{clm:B-weihgt} $w(B) \leq \epsilon^{-O(d)}\cdot w(\mst)$.
\end{claim}
\begin{proof}
We have:
\begin{equation}
\begin{split}
\frac{4g\dk}{\epsilon} \sum_{i} \ell_i &\leq \frac{4g\dk}{\epsilon} \ell_{\max} \sum_{i}\epsilon^i,\ \mbox{where $\ell_{\max} = \max_{e\in S}\{w(e)\}$} \\
    &\leq \frac{4g\dk}{\epsilon} w(\mst)  \sum_{i} \epsilon^i \\
    & \leq \frac{4g\dk}{\epsilon}  w(\mst) \frac{1}{1-\epsilon} =  \ed\cdot w(\mst)
\end{split}
\end{equation}
\end{proof}
 \bibliographystyle{plain}
\bibliography{spanner}

\appendix
\section{Notation and definitions} \label{sec:prel}

Let $G(V(G),E(G))$ be a connected and undirected graph with a positive edge weight function $w : E(G) \rightarrow \Re^+\setminus \{0\}$. We denote $|V(G)|$ and $|E(G)|$ by $n$ and $m$, respectively. Let $\MST(G)$ be a minimum spanning tree of $G$; when the graph is clear from the context, we simply write $\mst$.  A walk of length $p$ is a sequence of alternating vertices and edges $\{v_0,e_0,v_1,e_1,\ldots, e_{p-1}, v_{p}\}$ such that $e_i = v_iv_{i+1}$ for every $i$ such that $1\leq 0 \leq p-1$. A path is a \emph{simple walk} where every vertex appears exactly once in the walk. 
For two vertices $x,y$ of $G$, we use $d_G(x,y)$ to denote the shortest distance between $x$ and $y$. 

Let $S$ be a subgraph of $G$. We define $w(S) = \sum_{e \in E(S)}w(e)$. Let $X \subseteq V(G)$ be a set of vertices. We use $G[X]$ to denote the subgraph of $G$ induced by $X$. Let $Y\subseteq E(G)$ be a subset of edges of $G$. We denote the graph with vertex set $V(G)$ and edge set $Y$ by  $G[Y]$.

\section{Greedy spanners} \label{app:greed}

A subgraph $S$ of $G$ is a $(1+\epsilon)$-spanner of $G$ if $V(S) = V(G)$ and  $d_S(x,y) \leq (1+\epsilon)d_G(x,y)$ for all $x,y\in V(G)$. The following greedy algorithm by Alth\"ofer et al.~\cite{ADDJS93} finds a $(1+\epsilon)$-spanner of $G$:

\begin{tabbing}
  {\sc GreedySpanner}$(G(V,E), \epsilon)$\\
  \qquad \= $S \leftarrow (V, \emptyset)$.\\
  \> Sort edges of $E$ in non-decreasing order of weights.\\
  \> For each edge $xy \in E$ in sorted order\\
  \> \qquad \=  if $(1+\epsilon)w(xy) < d_S(x,y)$\\
  \>\>\qquad \=  $E(S) \leftarrow E(S) \cup \{e\}$\\
  \qquad \= return $S$
\end{tabbing}

\noindent Observe that as algorithm {\sc GreedySpanner} is a relaxation of  Kruskal's algorithm, $\MST(G) = \MST(S)$. Since we only consider $(1+\epsilon)$-spanners in this work, we simply call an $(1 + \epsilon)$-spanner a \emph{a spanner}. We define the lightness of a spanner $S$ to be the ratio $\frac{w(S)}{w(\MST(G))}$. We call $S$ \emph{light} if its lightness is independent of the number of vertices or edges of $G$.

\section{Omitted Proofs}\label{app:ommitted}
\begin{proof}[Proof of Lemma~\ref{lm:K-structure}]
To show that $\mathcal{K}$ is simple, we use the same argument as Borradaile, Le and Wulff-Nilsen~\cite{BLW17} that we briefly sketch here. Recall $\epsilon$-clusters have diameter at most $g\epsilon \ell$ by the diameter-credit invariants. Recall edges in $E_i$ have weight in range $(\ell/2,\ell]$. Thus, when $\epsilon$ is sufficiently small, $\mathcal{K}$ has no self-loops. To show that $\mathcal{K}$ has no parallel edges, we assume there are such two $xy$ and $uv$ where $w(xy) < w(uv)$. Let $C_u$, $C_v$ be two $\epsilon$-clusters that contain $u$ and $v$, respectively. We further assume,~w.l.o.g, that $x \in C_u, y\in C_v$. Then the $u$-to-$v$ path $P_{uv}$ from $u$ to $x$ inside $C_u$, edge $xy$ and then $y$ to $v$ inside $C_v$ has length at most $w(xy)  +2g\epsilon \ell$, which is at most $(1+4g\epsilon)w(uv)$ since $w(uv) \geq \ell/2$. Thus, by choosing $s \geq 4g$, edge $uv$ is not added to the spanner  by the greedy algorithm. Thus, $\mathcal{K}$ is simple.
\end{proof}

\end{document}